\documentclass[final,copyright,noderivs]{eptcs}
\providecommand{\event}{NCMA 2024} 
\usepackage{graphicx} 
\usepackage{mathrsfs}
\usepackage{amssymb}
\usepackage{amsmath}
\usepackage{amsthm}

\theoremstyle{plain}
\newtheorem{theorem}{Theorem}
\newtheorem{proposition}[theorem]{Proposition}
\newtheorem{lemma}[theorem]{Lemma}

\theoremstyle{definition}

\DeclareMathOperator{\lcm}{lcm}
\newcommand{\dfa}{\textrm{DFA}}
\newcommand{\ufa}{\textrm{UFA}}
\newcommand{\nfa}{\textrm{NFA}}
\newcommand{\xnfa}{\textrm{XNFA}}
\newcommand{\afa}{\textrm{AFA}}

\newcommand{\eoe}{\ifmmode$\hspace*{\fill}$\blacksquare\else\hspace*{\fill}$\blacksquare$\fi\smallskip}
\title{Complexity of Unary Exclusive\\ Nondeterministic Finite Automata}

\author{Martin Kutrib \and Andreas Malcher \and Matthias Wendlandt
\institute{%
  Institut f\"ur Informatik, Universit\"at Giessen\\
  Arndtstr.~2, 35392 Giessen, Germany}
\email{$\{$kutrib, andreas.malcher, matthias.wendlandt$\}$@informatik.uni-giessen.de}
}
 
\def\titlerunning{Complexity of Unary Exclusive NFAs}
\def\authorrunning{M.~Kutrib, A.~Malcher, M.~Wendlandt}

\begin{document}

\maketitle

\begin{abstract}
Exclusive nondeterministic finite automata ($\xnfa$) are nondeterministic
finite automata with a special acceptance condition.
An input is accepted if there is exactly one accepting path in its
computation tree. If there are none or more than one accepting paths,
the input is rejected. We study the descriptional complexity 
of $\xnfa$ accepting unary languages. While the state costs for 
mutual simulations with $\dfa$ and $\nfa$ over general alphabets 
differ significantly from the known types of finite 
automata, it turns out that the state costs for the simulations
in the unary case are in the order of magnitude of the general case.
In particular, the state costs for the simulation of an 
$\xnfa$ by a $\dfa$ or an $\nfa$ are $e^{\Theta(\sqrt{n\cdot\ln n})}$.
Conversely, converting an $\nfa$ to an equivalent $\xnfa$ may
cost $e^{\Theta(\sqrt{n\cdot\ln n})}$ states as well. 
All bounds obtained are also tight in the order of magnitude. 
Finally, we investigate the computational complexity of different decision
problems for unary $\xnfa$s and it is shown that the problems of emptiness, universality, inclusion, 
and equivalence are {\sf coNP}-complete, whereas the general membership problem is 
{\sf NL}-complete.
%
\end{abstract}

\section{Introduction}

The ability of using nondeterminism for finite automata does not increase their computational power 
in comparison with the deterministic variant, but the simulation costs for a deterministic finite automaton ($\dfa$)
can be exponentially higher in terms of states than for an equivalent nondeterministic finite 
automaton ($\nfa$)~\cite{Meyer:1971:edagfs,moore:1971:bssspe}.

In the last decades several structural extensions of finite automata have been examined.
One such extension is, for example, to give the reading head of the finite automaton the power of two-way motion.
Such \emph{two-way} finite automata do also not increase the computational power of
finite automata~\cite{Rabin:1959:fadp}, but they are interesting from a descriptional complexity point of
view, since the costs for one-way deterministic finite automata for the simulation of two-way deterministic finite automata 
can be exponential in the number of states~\cite{moore:1971:bssspe}. Similar results can also be shown 
for the nondeterministic case \cite{Sakoda:1978:NST}.

A more fine-grained look on the range between nondeterministic and deterministic finite automata leads to the 
model of \emph{unambiguous} finite automata~\cite{schmidt:1978:sdcfrfl}. Here, nondeterminism is allowed, 
but for every accepted word there has to be exactly one accepting path. 
From a descriptional complexity perspective it is known that 
the trade-off from unambiguous finite automata to $\dfa$s is exponential as \mbox{well 
\cite{leung:1998:seafapafa,leung:2005:dcnfada,schmidt:1978:sdcfrfl}.}

In contrast to these structural extensions, another extension is examined in~\cite{kutrib:2023:coenfa:proc,kutrib:2023:coenfa} 
that is based on the acceptance conditions of the automata and which leads 
to \emph{exclusive} nondeterministic finite automata (\xnfa). 
In this model, the computation tree of an input is defined in the same way as for nondeterministic finite 
automata, but its interpretation is different. Namely, an input word~$w$ is accepted, if there is 
exactly one accepting path for~$w$. If there is no accepting path for~$w$ or two or more accepting paths for~$w$, 
then~$w$ is rejected. Clearly, any unambiguous finite automaton can be considered as an $\xnfa$, but in comparison 
to unambiguous finite automata, multiple accepting paths are allowed and lead to non-acceptance in an $\xnfa$.
In~\cite{kutrib:2023:coenfa:proc,kutrib:2023:coenfa} complexity aspects of $\xnfa$s have been investigated.
Concerning the descriptional complexity, it is shown that $n$-state $\xnfa$s can be determinized 
as well, but the upper bound turns out to be $3^n-2^n+1$ and is shown to be tight. Moreover, 
$n\cdot 2^{n-1}$ states are shown to be a tight bound for the simulation of an $\xnfa$ by an equivalent $\nfa$.
The simulation of an $\nfa$ by an equivalent $\xnfa$ leads to an upper bound of $2^n-1$ which is shown
to be tight as well. Concerning the computational complexity, it is shown that the problems of emptiness, 
universality, inclusion, and equivalence are {\sf PSPACE}-complete, whereas the general membership problem is 
{\sf NL}-complete. It should be noted that a computational model with exactly one accepting computation 
on every accepted input has already been known in the context of complexity theory as the class \textsf{US}
(unique solution). It is defined (see~\cite{Blass:1982:otusp}) 
as the class of languages~$L$ for which there 
exists a nondeterministic polynomial time Turing machine $M$ such that
$w \in L$ if and only if~$M$ has on input $w$ exactly one accepting computation path.
A short overview on the properties of the class \textsf{US} may be found in~\cite{Hemaspaandra:2002:tctc}.

In this paper, we investigate the descriptional and computational complexity of $\xnfa$s accepting \emph{unary} languages. 
The descriptional complexity of unary regular languages has extensively been studied in the literature.
A fundamental result was obtained by Chrobak in~\cite{Chrobak:1986:FAUL,Chrobak:2003:ERRFAUL}. 
He shows that $O(F(n))$ is a tight bound for the simulation of an $\nfa$ by an equivalent $\dfa$.
Here, $F(n)$ denotes Landau's function~\cite{Landau:1903:mpgg} that is the maximal order of the cyclic subgroups of the 
symmetric group on~$n$ elements and can be estimated as $F(n) \in e^{\Theta(\sqrt{n\cdot\ln n})}$.
Landau's function plays a crucial role in many results on the descriptional complexity of unary regular languages.
One line of research in the past years is that many automata models such as, for example, one-way finite automata, 
two-way finite automata, pushdown automata, and context-free grammars have been investigated and compared
to each other with respect to simulation results and the size costs of the simulation 
(see, for example,~\cite{geffert:2003:ctwnuaisa,Mereghetti:2001:osbua,okhotin:2012:ufau,pighizzini:2009:dpdaul:art,pighizzini:2002:ucfgpdadcaslb}). 
Another line of research in recent years concerns investigations on the state complexity of operations on unary languages
which can be found, for example, in~\mbox{\cite{holzer:2003:ulonsc:proc,kunc:2012:scotwfaua,mera:2005:cunfa,Pighizzini:2002:uloscjf}}. 

The paper is structured as follows. In Section~\ref{sect:prelim}, we give the basic definitions that
are used in the further sections. In Section~\ref{sect:det}, we study the descriptional costs for determinizing
a given unary $\xnfa$. As a fundamental preparatory step we show that any unary $n$-state $\xnfa$ can be converted
to an equivalent $O(n^3)$-state $\xnfa$ in Chrobak normal form. This result is in slight contrast
to $\nfa$s where the conversion of an arbitrary $\nfa$ to Chrobak normal form may induce only a quadratic blow-up
of the number of states. Based on the $\xnfa$ in Chrobak normal form we can construct an equivalent
$\dfa$ whose number of states is bounded by $e^{\Theta(\sqrt{n\cdot\ln n})}$. This upper bound is also
tight in the order of magnitude. In Section~\ref{sect:nondet}, we obtain similar upper and lower bounds for
the conversion of unary $\xnfa$s to equivalent $\nfa$s and of unary $\nfa$s to equivalent $\xnfa$s.
Finally, in Section~\ref{sect:compcomp} we study the computational complexity
of decidability questions. In particular, we consider general membership, emptiness, universality, inclusion,
and equivalence with respect to the unary case and show
that for unary $\xnfa$s the general membership problem 
is {\sf NL}-complete, whereas the questions of emptiness, finiteness,
inclusion, and 
equivalence are {\sf coNP}-complete.

\section{Definitions and Preliminaries}\label{sect:prelim}

Let $\Sigma^*$ denote the set of all words over the finite alphabet $\Sigma$.
The \emph{empty word} is denoted by $\lambda$, and
$\Sigma^+ = \Sigma^* \setminus \{\lambda\}$. 
The \emph{reversal} of a word $w$ is denoted by $w^R$.  For the \emph{length} of~$w$ we
write~$|w|$. 
We use $\subseteq$ for \emph{inclusions} and~$\subset$ for \emph{strict inclusions}.
We write~$2^{S}$ for the power set and~$|S|$ for the cardinality of a
set~$S$. 

A \emph{nondeterministic finite automaton} ($\nfa$) is a
system $M=\langle Q,\Sigma,\delta,q_0,F\rangle$, where 
$Q$ is the finite set of \emph{states},
$\Sigma$ is the finite set of \emph{input symbols},
$q_0 \in Q$ is the \emph{initial state},
$F\subseteq Q$ is the set of \emph{accepting states}, and
$\delta\colon Q \times\Sigma\to 2^{Q}$ is the \emph{transition function}.  

With an eye towards further modes of acceptance, 
we define the \emph{acceptance of an input} in terms of computation trees. 
For any input $w=a_1a_2\cdots a_n\in\Sigma^*$ read by some $\nfa$ $M$, a
\emph{(complete) path for $w$} is a sequence of states $q_0,q_1,\dots,q_{n}$ such that
$q_{i+1}\in\delta(q_i,a_{i+1})$, $0\leq i\leq n-1$.
All possible paths on $w$ are combined into a \emph{computation tree} of
$M$ on $w$. So, a computation tree of $M$ is a finite rooted tree
whose nodes are labeled with states of $M$. In particular, the root is labeled
with the initial state, and the successor nodes of a node labeled $q$ are
the nodes $p_1,p_2,\dots, p_m$ if and only if $\delta(q,a)=\{p_1,p_2,\dots,
p_m\}$, for the current input symbol $a$.
A path in the computation tree is an
\emph{accepting path} if it ends in an accepting state.

Now, an input $w$ is accepted by an $\nfa$ 
if at least one path in the computation tree of $w$ is accepting.

An $\nfa$, where for acceptance
it is required that \emph{exactly} one path is accepting,
is called an \emph{exclusive nondeterministic finite
automaton} ($\xnfa$).

The \emph{language accepted} by the $\xnfa$~$M$ is 
$L(M) = \{\,w\in \Sigma^*\mid w \text{ is accepted by } M\,\}$.

Finally, an $\nfa$ is a \emph{deterministic
finite automaton} ($\dfa$) if and only if $|\delta(q,a)|=1$, for all $q\in Q$ 
and $a\in\Sigma$. In this case we
simply write $\delta(q,a)=p$ for $\delta(q,a)=\{p\}$ assuming that the
transition function is a mapping $\delta\colon Q \times\Sigma\to Q$. 
So, any $\dfa$ is complete, that is, the transition function is
total, whereas for the other automata types it is possible that $\delta$ maps to the 
empty set. 
A finite automaton is called \emph{unary} if its set of 
input symbols is a singleton. In this case we use $\Sigma=\{a\}$ throughout 
the paper.

\section{Determinization of unary XNFAs}\label{sect:det}

The problem of evaluating the costs of unary automata simulations was raised
in~\cite{Sipser:1980:lpssa}, and has led to emphasize some relevant
differences with the general case. For example, unary $\nfa$s can be much 
more concise than $\dfa$s, but yet not as much as for the general case.
Moreover, the sophisticated studies in~\cite{Mereghetti:2001:osbua}
reveal tight bounds for many other types of unary finite automata conversions.
The paper and the survey~\cite{pighizzini:2015:ialoaua} are also a
valuable source for further references.

For state complexity issues of unary finite automata, Landau's function
$$
F(n)=\max \{\,\lcm(c_1,c_2\dots,c_l) \mid l\geq 1, c_1,c_2,\dots,c_l \geq 1, 
c_1+c_2+\cdots + c_l=n\,\}
$$
which gives 
the maximal order of the cyclic subgroups of the symmetric group
on~$n$ elements, plays a crucial role, where $\lcm$ denotes the 
\emph{least common multiple}~\cite{Landau:1903:mpgg,Landau:1909:hlvp:book}.
It is well known that the $c_i$ always can be chosen to be relatively
prime. Moreover, an easy consequence of the definition is that
the~$c_i$ always can be chosen such that $c_1,c_2, \dots, c_l \geq 2$,
$c_1+c_2+\cdots +c_l \leq n$, and \mbox{$\lcm(c_1,c_2,\dots, c_l) = F(n)$}
(cf.,~for example,~\cite{Nicolas:1968:ssn}). 

Since~$F$ depends on the irregular distribution of the prime numbers we cannot
expect to express~$F(n)$ explicitly by $n$. 
In~\cite{Landau:1903:mpgg,Landau:1909:hlvp:book}
the asymptotic growth rate $\lim_{n\to\infty} (\ln
F(n)/\sqrt{n\cdot\ln n})=1$ was determined, which for our purposes
implies the (sufficient) rough estimate $F(n) \in
e^{\Theta(\sqrt{n\cdot\ln n})}$
(see also~\cite{ellul:2004:dcmrl,Szalay:1980:mos} for
bounds on~$F$). 
\begin{sloppypar}
The asymptotically tight bound of $F(n)$ for the unary NFA-to-DFA
conversion was presented \mbox{in~\cite{Chrobak:1986:FAUL,Chrobak:2003:ERRFAUL}.}
The proof is based on a normal form for unary $\nfa$s derived 
in~\cite{Chrobak:1986:FAUL}. Each $n$-state unary $\nfa$ can
effectively be converted into an equivalent $O(n^2)$-state $\nfa$ 
in this so-called Chrobak normal form.
However, the original proof in~\cite{Chrobak:1986:FAUL} 
contains an error that has been discovered and fixed 
in~\cite{to:2009:ufaap}. While the correction increases the state costs,
their order of magnitude is not affected.
In connection with magic numbers, more precise and improved state bounds have
been shown in~\cite{geffert:2007:mnshfa} by a completely different proof.
\end{sloppypar}

Let $t,d\geq 0$ be two integers. An \emph{arithmetic progression with 
offset $t$ and period $d$} is the set 
$$
\{\, t + x\cdot d\mid x\geq 0\,\}.
$$

We recall a well-known useful fact 
which is related to number theory and Frobenius numbers 
(see, for example,~\cite{shallit:2008:fpaig} for a survey).

\begin{lemma}\label{lem:all-unary-words}
Let $0 < c_1 < c_2 < \cdots < c_r\leq n$ be positive integers. 
Then the set of integers $z > n^2$ that can be written as a non-negative 
integer linear combination of the $c_i$ 
is $\{\, t + x\cdot d\mid x\geq 0\,\}$, where
$t$ is the least integer greater than $n^2$ that is a 
multiple of $d = \gcd(c_1,c_2,\dots, c_r)$.
\end{lemma}

\begin{sloppypar}
A unary $\xnfa$ $M=\langle Q,\{a\},\delta,q_0,F\rangle$
is in \emph{Chrobak normal form} if, for some $m\geq 0$ and \mbox{$k\geq 0$,} 
\mbox{$Q = \{\, q_i \mid 0\leq i\leq m\,\} \cup C_1 \cup C_2 \cup \cdots \cup C_k$,}
where, for each $1\leq i\leq k$,
$C_i=\{\, p_{i,0}, p_{i,1},\dots ,p_{i,j_{i}-1}\,\}$ for some $j_i\geq 1$,
$\delta(q_i,a)=\{q_{i+1}\}$ for \mbox{$0\leq i\leq m-1$,}
and for each $1\leq i\leq k$ and $0\leq h\leq j_i-1$,
\mbox{$\delta(p_{i,h},a)=\{p_{i,(h+1)\bmod j_i}\}$,} and
\mbox{$\delta(q_m,a)=\{p_{1,0}, p_{2,0},\dots, p_{k,0}\}$.}
\end{sloppypar}

So, an $\xnfa$ is in Chrobak normal form if its structure is a deterministic
tail from $q_0$ to $q_m$, where the automaton makes only a single
nondeterministic decision, which chooses one of the disjoint cycles~$C_i$.

Next, we show how to convert a unary $\xnfa$ into Chrobak normal form. The idea of
the construction is along the lines of the construction
in~\cite{to:2009:ufaap} but with modifications with respect to the
exclusiveness of the $\xnfa$.

\begin{lemma}\label{lem:chrobak}
Let $n\geq 1$. For every unary $n$-state $\xnfa$, 
an equivalent $O(n^3)$-state $\xnfa$ in Chrobak normal form 
can effectively be constructed, such that the sum of the cycle 
lengths is of order $O(n)$.
\end{lemma}

\begin{proof}
Let $M=\langle Q,\Sigma,\delta,q_0,F\rangle$ be an $n$-state $\xnfa$.
Since any unary language over some alphabet is
completely determined by the lengths of the words in the language,
we can safely disregard $\Sigma$ and 
consider the state graph of~$M$ only. For $L(M)=\emptyset$, the theorem
is trivial. So, in the sequel we assume that~$L(M)$ is not empty.
Moreover, we may safely assume that all states $q\in Q$ are reachable and
productive, that is, there is a path from $q_0$ to $q$ and a path from
$q$ to a final state. Now, by adding states and possibly removing some states 
and transitions, we modify $M$ such that there is no incoming transition to 
the initial state, such that $F = \{q_+\}$ is a singleton, and such that $q_+$ is 
the only state without outgoing transitions. 
To this end, all unreachable states together with their incoming and outgoing transitions
are removed. Similarly, all unproductive states together with their incoming
and outgoing transitions
are removed as well. Next, if the initial state has incoming transitions, a
new state without incoming transitions is added whose outgoing transitions go to the successor
states of the initial state. This new state becomes the new initial state.
In order to make $F$ a singleton, we have to take care
about words that are accepted on more than one path. So, first 
a new accepting state~$q_+$ is added. 
For each pair of old accepting states, if both states do not share a common
predecessor state, from each of their predecessor states a transition 
to~$q_+$ is added. Both states become non-accepting. However, if both states have
at least one common predecessor, say $p$, then there are two paths via $p$ to
accepting states. This means that inputs following these paths do not belong
to $L(M)$. In this case, both states become non-accepting, some state $p'$ is
added, and all incoming transitions to $p$ are doubled and are directed 
to~$p'$ as well. Furthermore, a transition from $p$ to $q_+$
and a transition from $p'$ to $q_+$ is added. Similarly, for all common
predecessors of the old accepting states. In this way, we obtain an $\xnfa$
equivalent to $M$ that has the desired properties. For convenience, we call
it also~$M$. The modified $\xnfa$ has at most $m=2n$ states.

From now on, we identify $M$ with its state graph.
Let $S$ be the set of non-trivial strongly connected components of $M$.
A \emph{superpath} in $M$ is a subgraph
$$
\alpha = P_1 S_1 P_2 S_2 \cdots P_\ell S_{\ell} P_{\ell+1},
$$ 
where, 
for $1\leq i \leq \ell$, $S_i \in S$;
for $1\leq i \leq \ell+1$, $P_i$ is a path in $M$ whose inner nodes do not
belong to non-trivial strongly component components of $M$;
the first node of~$P_1$ is~$q_0$; 
the last node of $P_{\ell+1}$ is $q_+$; 
for $1\leq i \leq \ell$, the last node of~$P_i$ belongs to~$S_i$;
for $2\leq i \leq \ell+1$, the first node of~$P_i$ belongs to~$S_{i-1}$.

For every superpath $\alpha$ in $M$, let $L_\alpha$ be the set of all
lengths of paths in $M$ from $q_0$ to $q_+$ that are in $\alpha$. 
It follows that the length of any accepting path in~$M$ belongs to
$\bigcup_\alpha L_\alpha$, where the union 
ranges over all superpaths in $M$.

We define the set $\Psi_\alpha$ to be the subset of paths from $q_0$ to $q_+$ in
$\alpha$ that are simple, that is, no state appears twice. Clearly, the
length of any path in $\Psi_\alpha$ does not exceed~$m$.

Next, we define $\Pi_\alpha$ to be another subset of paths from $q_0$ to $q_+$ in
$\alpha$. In particular, for every path~$\sigma$ in $\Psi_\alpha$,
we put the following extensions $\sigma'$ of $\sigma$ into~$\Pi_\alpha$.
Whenever $\sigma$ enters a strongly connected component $S_i$ in some 
state~$v$, then a Hamiltonian walk in~$S_i$ (that is, a tour that
visits all nodes in $S_i$) that cannot be shortened and that starts and ends 
in $v$ is inserted into $\sigma$.
Note that a  Hamiltonian walk that cannot be shortened is a path from which no
nodes can be removed without obtaining a path that is no longer
Hamiltonian. It needs not to be the shortest Hamiltonian walk in $S_i$.
Since $S_i$ is strongly connected, such Hamiltonian walks exist.
Results in~\cite{hamidoune:1979:slphdlgo}
show that the lengths of such Hamiltonian walks in $S_i$ do not
exceed~$|S_i|^2$, 
where  $|S_i|$ denotes the number of nodes in~$S_i$.
Therefore, the length of any path in $\Pi_\alpha$ does not exceed
$m^2+m$.

Now we consider a fixed superpath $\alpha$ in $M$. 
Let \mbox{$0 < c_1 < c_2 < \cdots < c_r \leq m$} be the lengths of all 
simple cycles in $\alpha$, $\sigma$ in $\Psi_\alpha$, and
$\sigma'$ be an extension of~$\sigma$ in~$\Pi_\alpha$.
Since $\sigma'$ visits each node in $\alpha$ at least once, the set 
$Z_{\alpha,\sigma'}$ of all lengths $z$ for which
$z= |\sigma'| + x_1 c_1 + x_2 c_2 + \cdots + x_r c_r$ is solvable 
in non-negative integers is contained in~$L_\alpha$.
By Lemma~\ref{lem:all-unary-words},
\mbox{$Z_{\alpha,\sigma'} = X_\alpha \cup \{\, t_{\sigma'} + x\cdot d\mid x\geq 0\,\}$,} 
where~$X_\alpha$ contains lengths not larger than $2m^2+m$ and
$t_{\sigma'}$ is the least integer greater than~\mbox{$2m^2+m$} such that
$t_{\sigma'} \equiv |\sigma'|\; (\bmod\; d)$, where 
$d = \gcd(c_1,c_2,\dots, c_r)$. Since the Hamiltonian walks in~$\sigma'$
are (compound) cycles, that is, linear combinations of
$c_1, c_2, \dots, c_r$, the number $d$ divides their lengths and, thus,
we have \mbox{$t_{\sigma'} \equiv |\sigma|\; (\bmod\; d)$.}

On the other hand, the set of all lengths $y$ for which there is a
$\sigma$ in $\Psi_\alpha$ such that
$$
y= |\sigma| + x_1 c_1 + x_2 c_2 + \cdots + x_r c_r
$$ 
is solvable 
in non-negative integers, clearly contains $L_\alpha$.
Therefore, if \mbox{$w\in L_\alpha$} and $w > 2m^2+m$ then
Lemma~\ref{lem:all-unary-words} implies that
there is a
$\sigma$ in $\Psi_\alpha$ such that
$w \equiv |\sigma|\; (\bmod\; d)$.
Since 
$\{\, t_{\sigma} + x\cdot d\mid x\geq 0\,\}\subseteq Z_{\alpha,\sigma'}$,
we conclude $w\in Z_{\alpha,\sigma'}$.

Altogether, we have $L_\alpha = N_\alpha \cup 
\bigcup_{\sigma'\in \Pi_\alpha} \{\, t_{\sigma'} + x\cdot d\mid x\geq 0\,\}$,
where $N_\alpha$ contains lengths not larger than $2m^2+m$.

So far, we have created the prerequisites for constructing the 
normal form without specifically addressing $\xnfa$s. So, the next task
is to assemble an $\xnfa$ $M'=\langle Q',\{a\},\delta',q'_0,F'\rangle$
equivalent to $M$ in Chrobak normal form.

To this end, we start with a deterministic tail consisting of the $m^3+2$ states
$\{\, q'_i \mid 0\leq i\leq m^3+1\,\}$ with
$\delta'(q'_i,a)=\{q'_{i+1}\}$, for $0\leq i\leq m^3$.
A state $q_i$ of the tail becomes accepting if and only if
the input of length $i$ belongs to $L(M)$. So, all words whose length 
does not exceed $m^3+1$ are correctly accepted or rejected.

Next, we want to add the cycles to the initial tail of $M'$.

To construct the cycles appropriately, we consider each superpath~$\alpha$ 
of~$M$ and distinguish three cases, respectively. As before, let 
\mbox{$0< c_1 < c_2 <\cdots < c_r\leq m$} be the lengths of all 
simple cycles in $\alpha$ and \mbox{$d= \gcd(c_1,c_2,\dots, c_r)$.}
We consider all inputs of lengths \mbox{$z > m^3+1 \in L_\alpha$.}

Case 1: There are at least two simple cycles $C_1$ and $C_2$ in $\alpha$.
Then, each path of length $z$ in $\alpha$
that can be shortened to some path in $\Pi_\alpha$ by deleting cycles, sees 
at least $z-(m^2+m)$ nodes in complete simple cycles of $\alpha$. 
If one of these paths contains at least two
different cycles of the same length, then these cycles can replace each other
and, thus, there are at least two accepting paths of length~$z$ in $\alpha$.
Therefore, the input of length $z$ does not belong to $L(M)$. Assume now that
all cycles in these paths have different lengths. Then there are at most $m$
cycles. Assume that each of these cycles is passed through at most $m-1$
times. Then,
\begin{multline*}
z\leq m^2+m + \sum_{i=1}^{m} i (m-1) = m^2+m + \frac{m^2+m}{2} (m-1)\\
= \frac{m^2+m}{2} (m+1)
= \frac{m^3+2m^2+m}{2}
\leq m^3+1 < z.
\end{multline*}
From the contradiction we conclude that there is at least one cycle, say~$C_1$,
that is passed through for $x_1\geq m$ times. Let $C_2$ be passed
through for $x_2$ times. We have $x_1\geq m \geq |C_2|\geq 1$ and $|C_1|\geq 1$.
So, $x_1 |C_1| + x_2 |C_2| = (x_1-|C_2|)|C_1| + (x_2 + |C_1|) |C_2|$.
The equality means that passing $x_1$ times through the cycle $C_1$ and $x_2$
times through the cycle $C_2$ is equivalent to passing $(x_1-|C_2|)$ times
through the cycle~$C_1$ and $(x_2 + |C_1|)$ times through the cycle $C_2$.
So, there are at least two accepting paths of length~$z$ in $\alpha$.
Therefore, the input of length $z$ does not belong to $L(M)$.

Case 2: There is exactly one simple cycle $C_1$ in $\alpha$.
So, there is at most one non-trivial strongly connected component in $\alpha$
and this strongly connected component is the cycle $C_1$. 
Clearly, in this case we have $d= |C_1|$ and the input length $z$ is uniquely
accepted along $\alpha$.

Case 3: There is no simple cycle in $\alpha$. 
In this case, there is no non-trivial strongly connected component in $\alpha$ and 
the unique path of length~$z$ from the initial state 
ends in the initial tail and, by construction, the input
of length $z$ is correctly accepted or rejected.

Now we are ready to add the cycles for $\alpha$ to the tail of $M'$. 
To this end, nothing has to be done for Case 3. 

\begin{sloppypar}
For the remaining cases, the cycle length must be $d$. 
If there is no cycle of length $d$, we add two disjoint
cycles $A_\alpha$ and $R_\alpha$ each of length $d$.
In particular, $A_\alpha$ consists of states 
$\{\, s_{0}, s_{1},\dots ,s_{d-1}\,\}$ with 
$\delta'(s_{h},a)=\{s_{(h+1)\bmod d}\}$, and similarly,~$R_\alpha$ consists of states 
$\{\, r_{0}, r_{1},\dots ,r_{d-1}\,\}$ with 
\mbox{$\delta'(r_{h},a)=\{r_{(h+1)\bmod d}\}$.} The cycles
are connected to the tail by the transitions
$\delta(q'_{m^3+1},a)=\{s_{0}\}$ and
$\delta(q'_{m^3+1},a)=\{r_{0}\}$.
If there are already two cycles 
$A$ and $R$ of length $d$ that have already been constructed
for some other superpath, then they are reused
and nothing is added.
\end{sloppypar}

Next, we identify the accepting states on the cycles.

For Case 1, we consider each $\sigma\in \Psi_\alpha$
and states $s_i$ and $r_i$ become accepting if 
$m^3+1+i+1 \equiv |\sigma|\; (\bmod\; d)$.
In this way, Case 1 is treated correctly, since
now two different paths in $M'$ are accepting 
for the same length.

For case 2, we also consider each $\sigma\in 
\Psi_\alpha$. Here, only state $s_i$ becomes accepting if 
$m^3+1+i+1 \equiv |\sigma|\; (\bmod\; d)$.

In this way, Case 2 is treated correctly, since
only one path is made accepting. However, it
may be that $r_i$ was already accepting. This means that
the corresponding inputs are also accepted by another
superpath.

This concludes the construction of $M'$. Note, if an input is
accepted by different superpaths having different cycle length,
then it clearly does not belong to~$L(M')$, but is also does 
not belong to $L(M)$. Conversely, if an input is accepted
unambiguously by $M$ then it is accepted also unambiguously
by $M'$. So, we conclude $L(M)=L(M')$. Moreover, since the sum
of the different cycle lengths is at most $m$ and each cycle length
appears at most twice, the total sum of the cycle lengths is at most $2m$. 
\end{proof}

Next, we can utilize the normal form to show that the costs for the
determinization of \emph{unary} $\xnfa$s are the same (in the order 
of magnitude) as for $\nfa$s. This is in strict contrast to $\xnfa$s
over a general alphabet. 
The backbone of the construction is similar to the backbone of the
construction given in~\cite{Chrobak:1986:FAUL}. However, here we 
have to treat the cases when inputs are accepted at multiple paths.

\begin{theorem}\label{theo:unary-xnfa-to-dfa} 
Let $n\geq 1$ and~$M$ be a unary $n$-state $\xnfa$.
Then $e^{\Theta(\sqrt{n\cdot\ln n})}$ states are sufficient 
for a $\dfa$ to accept~$L(M)$. 
\end{theorem}

\begin{proof}
Given a unary $n$-state $\xnfa$ $M$, we first construct an
equivalent $O(n^3)$-state $\xnfa$ $M'$ in Chrobak normal form 
as in the proof of Lemma~\ref{lem:chrobak}.
Let $A_1,R_1,A_2,R_2,\dots , A_k,R_k$, for $k\geq 1$, be the cycles
of $M'$, where \mbox{$|A_i|=|R_i|$,} for $1\leq i\leq k$.
We construct the equivalent 
$\dfa$ $M''= \langle Q,\Sigma,\delta,q_0,F\rangle$ as follows.

\begin{sloppypar}
First, we take over the initial deterministic tail of $M'$, which 
has the $m^3+2$ states \mbox{$\{\, q'_i \mid 0\leq i\leq m^3+1\,\}$,}
where $m=2n$ as in the proof of Lemma~\ref{lem:chrobak}.
Then we add one big cycle of length
\mbox{$\ell = \lcm\{|A_1|,|A_2|,\dots , |A_k|\}$} to the tail. 
To this end the states from the set \mbox{$\{\, p_i \mid 0\leq i\leq \ell-1\,\}$}
are cyclically connected and a transition from $q_{m^3+1}$ to $p_0$ is added. 
\end{sloppypar}

Next, we have to identify the accepting states. To this end, all accepting
states on the tail remain accepting. So, as for $M'$ all words up to length
$m^3+1$ are treated correctly.

Then, we assume that each state $p_i$ of the cycle has a counter attached that 
is initially set to $0$. Now, we consider each cycle $A_i$ of $M'$
consisting of the states $\{s_{0}, s_{1},\dots ,s_{d-1}\}$.
Whenever a state $s_j$ is accepting, then the counters of all
states $\{\, p_t \mid t=j+ x\cdot d, \text{ for } 0\leq x \leq \frac{\ell}{d}-1\,\}$
are increased by one. Similarly, for each cycle $R_i$ of $M'$
consisting of the states $\{r_{0}, r_{1},\dots ,r_{d-1}\}$.
If a state~$r_j$ is accepting, then the counters of all
states \mbox{$\{\, p_t \mid t=j+ x\cdot d, \text{ for } 0\leq x \leq \frac{\ell}{d}-1\,\}$}
are increased by one.

In a last construction step, all states whose counters are exactly one
become accepting, all the others become non-accepting. In this way, all inputs
that are accepted by more than one path in $M'$ are rejected in~$M''$, and all
inputs that are accepted in $M'$ and, thus, in $M$ by exactly one path are
accepted by~$M''$ as well. So, $L(M)=L(M'')$ and, clearly, $M''$ is a $\dfa$.
Moreover, $M''$ has at most 
$$
m^3+2+\ell \leq (2n)^3+2+\ell\leq (2n)^3+2+F(n)
\in e^{\Theta(\sqrt{n\cdot\ln n})}
$$ 
many states.
\end{proof}

It will turn out after Proposition~\ref{prop:xnfanfalower:unary2} that 
the upper bound for the determinization in
Theorem~\ref{theo:unary-xnfa-to-dfa} is tight in the order of magnitude.

\section{Converting unary NFAs to XNFAs and Vice Versa}\label{sect:nondet}

\begin{sloppypar}
Here, again Landau's function 
$$
F(n)=\max \{\,\lcm(c_1,c_2\dots,c_l) \mid l\geq 1, c_1,c_2,\dots,c_l \geq 1, 
c_1+c_2+\cdots + c_l=n\,\}
$$
plays a crucial role.
Recall that the $c_i$ always can be chosen to be relatively prime
such that \mbox{$c_1,c_2, \dots, c_l \geq 2$,}
$c_1+c_2+\cdots +c_l \leq n$, and \mbox{$\lcm(c_1,c_2,\dots, c_l) = F(n)$}.
This, for
example, means that the $c_i$ can be prime powers. An interesting and
simplifying result in~\cite{miller:1987:mog} revealed that, instead of
prime powers, one can sum up the first prime numbers such that the 
sum does not exceed the limit~$n$. More, precisely,
it has been shown in~\cite{miller:1987:mog} that the following function~$G(n)$ 
is of the same order of magnitude as $F(n)$, that is, $G(n)\in \Theta (F(n))$.
Let $p_i$ denote here the $i$th prime number with $p_1=2$.
$$
G(n)= \max \{\,p_1\cdot p_2\cdots  p_l \mid l\geq 1 \text{ and } p_1+p_2+\cdots + p_l\le n\,\}
$$
\end{sloppypar}

In the following theorem we use the function $G(n)$ to describe the worst case
state costs of an $\nfa$ simulating a unary~$\xnfa$.

\begin{theorem}\label{theo:xnfanfalower:unary} 
Let $n\geq 2$. There exists a unary $(n+1)$-state $\xnfa$ $M$ such that every 
$\nfa$ in Chrobak normal form accepting $L(M)$ has at least $G(n)$ states.
\end{theorem}

\begin{proof}
\begin{sloppypar}
For $n\geq 2$, let $G(n)$ be represented by the product $p_1\cdot p_2\cdots
p_l$ of the first $l\geq 1$ prime numbers.
We consider the $\xnfa$ $M=\langle Q,\{a\},\delta,q_0,F\rangle$ whose state
graph has $l$ disjoint cycles. Each cycle \mbox{$1\leq i\leq l$} has length $p_i$
and consists of the states 
$\{\, r_{i,0}, r_{i,1},\dots ,r_{i,p_i-1}\,\}$, where
\mbox{$\delta(r_{i,h},a)=\{r_{i,(h+1)\bmod p_i}\}$,} for \mbox{$0\leq h\leq p_i-1$.}
Now, the initial state $q_0$ is nondeterministically connected to the cycles
by \mbox{$\delta(q_0,a)=\{r_{1,1}, r_{2,1},\dots, r_{l,1}\}$.}
The set of accepting states is $F=\{\, r_{i,0}\mid 1\le i\le l\,\}$. 
By construction,~$M$ has at most $n+1$ states.
\end{sloppypar}

The language $L(M)$ accepted by $M$ is
$$
\{\,a^m \mid \text{there is exactly one } i\in\{1,2,\dots,\ell\} \text{ such
    that } m \equiv 0\; (\bmod\; p_i)\,\}.
$$

We define the set of all integers that are not divisible by all $p_i$, $1\leq
i\leq l$,
as 
$$
K=\{\, k\in \mathbb{N} \mid k \text{ is not divisible by all } p_i, 1\leq
i\leq l\,\}.
$$

Assume now, that $L(M)$ is accepted by an $\nfa$ $M'$ in Chrobak normal form with
less than $G(n)$ states, say $m< G(n)$ states. 

Our first goal is to show the claim that for any $p_i$, $1\leq i\leq l$, 
all cycles in the state graph of $M'$ on which infinitely
many words from $\{\,a^{x\cdot p_i} \mid x\in K \,\}$ are accepted,
have a length that is divisible by $p_i$.

Since all words from the infinite set $\{\,a^{x\cdot p_i} \mid x\in K \,\}$
belong to $L(M)$, cycles on which infinitely many such words are accepted
exist. Assume that one of these cycles has a length $c$ not divisible by $p_i$
and let~$a^{x_0\cdot p_i}$ with $x_0\in K$ be one of the accepted words.
Then, the word $w=a^{x_0\cdot p_i+ c\cdot p}$ with $p=\frac{G(n)}{p_i}$
is accepted as well.
But since $c$ and $p$ are not divisible by $p_i$, we have that $|w|$ is not
divisible by $p_i$, either. Moreover, since $x_0\cdot p_i$ is not divisible by any~$p_j$ with $i\neq j$
but~\mbox{$c\cdot p$} is, we have that $|w|$ is not divisible by any~$p_j$ with $i\neq j$, either.
So, $w$ cannot belong to $L(M')$. 
From this contradiction the claim follows.

Since $m< G(n)$, there must be two cycles $C_1$ and $C_2$, say of length $c_1$ and~$c_2$,
such that there are two different prime numbers $p_i\ne p_j$ with $1\leq i,j
\leq l$, where~$c_1$ is divisible by $p_i$ but not divisible by~$p_j$
and infinitely many words from $\{\,a^{x\cdot p_i} \mid x\in K \,\}$ are
accepted in $C_1$, and
where $c_2$ is divisible by~$p_j$ but not divisible by $p_i$
and infinitely many words from $\{\,a^{x\cdot p_j} \mid x\in K \,\}$ are
accepted in~$C_2$.
Since $p_j$ is relatively prime to $c_1$, there is an integer
$p$ such that \mbox{$p\cdot c_1 \equiv 1\; (\bmod\;  p_j)$.} 
Consider some word $w=a^{x_0\cdot p_i}$ with $x_0\in K$
that is accepted in $C_1$. 
Then, the word $a^{x_1\cdot p\cdot c_1+|w|}$ with $(x_1+|w|) \equiv 0\; (\bmod\; p_j)$
is accepted in $C_1$ as well. 
However, this word does not belong to $L(M)$, since it is
divisible by~$p_i$ and $p_j$. 

So, from this contradiction we conclude there is no $\nfa$ in
Chrobak normal form with less than $G(n)$ states. 
\end{proof}

Clearly the upper bound for the simulation of an $\xnfa$ by an $\nfa$ is given
by determinization. 
Thus, we have the following proposition.

\begin{proposition}\label{prop:xnfanfalower:unary2} 
Let $n\geq 2$ and~$M$ be a unary $n$-state $\xnfa$.
Then $e^{\Theta(\sqrt{n\cdot\ln n})}$ states are sufficient 
for an $\nfa$ to accept~$L(M)$. 
\end{proposition}

The lower bound in Theorem~\ref{theo:xnfanfalower:unary} 
says that there are $(n+1)$-state $\xnfa$s such that any equivalent
$\nfa$ \emph{in Chrobak normal form} has at least $G(n)$ states.
Moreover, any $n$-state $\nfa$ can be converted into an equivalent
$\nfa$ in Chrobak normal form that has at most $O(n^2)$ states.
So, since $G(n)\in \Theta(F(n))$~\cite{miller:1987:mog},
the lower bound for the state costs of the simulation of an 
$n$-state $\xnfa$ by an $\nfa$ (not necessarily in Chrobak normal form)
is
$$
\Theta(\sqrt{G(n-1)}) = \Theta(\sqrt{e^{\Theta(\sqrt{(n-1)\cdot\ln
      (n-1)})}})
= e^{\Theta(\sqrt{n\cdot\ln n})}.
$$
So, we conclude that the upper bound for the unary $\xnfa$-to-$\dfa$ conversion
shown in Theorem~\ref{theo:unary-xnfa-to-dfa} and the 
upper bound for the unary $\xnfa$-to-$\nfa$ conversion
shown in Proposition~\ref{prop:xnfanfalower:unary2}
are tight in the order of magnitude.

We turn to the simulation of $\nfa$s by $\xnfa$s.
In \cite{okhotin:2012:ufau} it has been shown that the language
$$
L=\{\, a^n\mid n\not\equiv 0\; (\bmod\; \lcm(c_1,c_2,\dots, c_k)) \,\}\cup \{\lambda\},
$$
for $k\ge 1$ and $c_1,c_2,\dots ,c_k\geq 2$ is accepted by
an $\nfa$ with $1+\sum_{i=1}^k c_i$ states, while the smallest $\ufa$ for~$L$
needs at least $1+\lcm (c_1,c_2,\dots ,c_k)$ many states.
The proof of the lower bound is based on a method given
in~\cite{schmidt:1978:sdcfrfl} which is based on a
rank argument on certain matrices. 
After a thorough analysis of the arguments of the method, it turned out that
exclusively accepting computations of the $\ufa$s are used. In other words, 
the arguments can be applied to $\xnfa$s as well.
So, we derive that also the smallest $\xnfa$ needs at
least $1+\lcm (c_1,c_2,\dots ,c_k)$ states to accept the language $L$. 
So, we have the following lower bound. 

\begin{theorem}\label{theo:nfaxnfalower:unary} 
Let $n\geq 2$. There exists a unary $(n+1)$-state $\nfa$ $M$ such that every 
$\xnfa$ accepting $L(M)$ has at least $F(n)+1$ states.
\end{theorem}

Clearly the upper bound for the simulation of an $\nfa$ by an $\xnfa$ is given
by determinization. 
Thus, we have the following proposition.

\begin{proposition}\label{prop:nfa-xnfa-lower:unary} 
Let $n\geq 2$ and~$M$ be a unary $n$-state $\nfa$.
Then $e^{\Theta(\sqrt{n\cdot\ln n})}$ states are sufficient 
for an $\xnfa$ to accept~$L(M)$. 
\end{proposition}

As before, we also conclude here that 
the lower bound and upper bound are tight in the order of magnitude.

\section{Computational Complexity}\label{sect:compcomp}

In this section, we discuss the computational complexity of
decidability questions. In particular, we consider general membership, emptiness, universality, inclusion,
and equivalence with respect to the unary case. These problems have been studied in~\cite{kutrib:2023:coenfa:proc,kutrib:2023:coenfa}
in case of general alphabets. It turns out here that the general membership problem
in the unary case shares the same computational complexity with the general case,
namely, both problems are {\sf NL}-complete.
However, the questions of emptiness, universality, inclusion, and equivalence
turn out to be {\sf coNP}-complete in the unary case, whereas these questions have
been shown to be {\sf PSPACE}-complete in the general case~\cite{kutrib:2023:coenfa:proc,kutrib:2023:coenfa}.

\begin{theorem}\label{theo:compmember:unary}
The problem of testing the general membership for unary $\xnfa$s 
is {\em \textsf{NL}}-complete.
\end{theorem}

\begin{proof}
To show that the problem is in {\sf NL} for unary $\xnfa$s we can use the same construction
that has been described in~\cite{kutrib:2023:coenfa:proc,kutrib:2023:coenfa} for general alphabets.
The basic idea is to test whether an input~$w$ is not accepted by a given $\xnfa$ $A$. This means
that either there is no accepting path in the computation tree for~$w$ or there are at least two
accepting paths. In the first case, the input~$w$ is not accepted by~$A$ even if~$A$ is considered as
an $\nfa$. Hence, this case can be solved in {\sf NL} using the known algorithms for $\nfa$s.
The second case can be checked by guessing two different accepting paths in the computation tree.
To this end, one has to keep track of two states representing the current position on the two paths.
Since this can be realized in {\sf NL}, the general membership problem is in~{\sf NL} in particular for 
unary $\xnfa$s.

To show the \textsf{NL}-hardness of the general membership problem for unary $\xnfa$s 
we can in principle apply the reduction that is described in~\cite{kutrib:2023:coenfa} for general alphabets.
To adapt it to the unary case we have to use the fact that the membership problem for unary $\nfa$s remains \textsf{NL}-complete 
(see, e.g.,~\cite{jones:1975:sbracp}) and we have to observe that the $\xnfa$ constructed in the reduction is unary, since the given $\nfa$ is unary.
Since the reduction described in~\cite{kutrib:2023:coenfa} is not yet published we provide the reduction here
for the sake of completeness.

To show the \textsf{NL}-hardness of the general membership problem 
we reduce the non-membership problem for $\nfa$s
which is known to be \textsf{NL}-complete, since
the membership problem for $\nfa$s is \textsf{NL}-complete.

Let $\langle A,w \rangle$ be the encoding of an $\nfa$
$A=\langle Q, \{a\}, \delta, q_0, F \rangle$ and an input word~$w$.
We construct an $\xnfa$ $A'=\langle Q \cup \{p_0,p\}, \{a\}, \delta', p_0, F' \rangle$,
where~$p_0$ and~$p$ are two new states not belonging to $Q$. The accepting states $F'$ are
defined as $F'=F \cup \{p_0,p\}$, if $\lambda \in L(A)$, and $F'=F \cup \{p\}$ otherwise.
The transition function $\delta'$ is defined as follows.
First, $A'$ has the same behavior as $A$ on states from~$Q$.
Formally, $\delta'(q,a)=\delta(q,a)$ for all $q \in Q$.
Second, from the new initial state $p_0$ all states are reached that are
reached from the initial state~$q_0$ of~$A$. Additionally, the new state~$p$
is reached from~$p_0$.
Formally, $q' \in \delta'(p_0,a)$, if $q' \in \delta(q_0,a)$,
and $p \in \delta'(p_0,a)$.
Finally, the state~$p$ acts as an accepting sink state, that is, $p \in \delta'(p,a)$.

The reduction from the encoding $\langle A,w \rangle$ to an encoding $\langle A',w \rangle$
can be realized by a deterministic logarithmically space-bounded Turing machine. 

For the correctness of the reduction we have to show that
the $\xnfa$ $A'$ accepts~$w$ if and only if $w$ is not accepted by the $\nfa$~$A$.
On the one hand, if~$w$ is accepted by~$A'$, then $p \in \delta'(p_0,w)$ and $\delta'(p_0,w) \cap F=\emptyset$,
since otherwise there would be at least two accepting paths for~$w$. Hence, $w$ is not
accepted by the $\nfa$~$A$.
On the other hand, if $w$ is not accepted by~$A'$, then $p \in \delta'(p_0,w)$ and 
$\delta'(p_0,w) \cap F \neq \emptyset$, since there must be at least two accepting paths for~$w$.
Hence,~$w$ is accepted by the $\nfa$~$A$.
This concludes the correctness of the reduction and shows the \textsf{NL}-hardness of
the general membership problem for $\xnfa$s. Altogether, we obtain that the
general membership problem for $\xnfa$s is \textsf{NL}-complete.
\end{proof}

It is known that the emptiness problem for unary $\nfa$s is {\sf NL}-complete. In contrast,
we show the problem becomes {\sf coNP}-complete for unary $\xnfa$s. In the following proofs we need a result
obtained in~\cite{kutrib:2023:coenfa:proc,kutrib:2023:coenfa} on the conversion of $\xnfa$s to $\dfa$s
in case of general alphabets.

\begin{theorem}\label{theo:xnfatodfa:upper}\cite{kutrib:2023:coenfa:proc,kutrib:2023:coenfa}
Let $n\geq 1$ and~$M$ be an $n$-state $\xnfa$.
Then $3^n-2^n+1$ states are sufficient for a $\dfa$ to accept~$L(M)$. 
\end{theorem}

\begin{theorem}\label{theo:compempty:unary} 
The emptiness problem for unary $\xnfa$s 
is {\em \textsf{coNP}}-complete.
\end{theorem}

\begin{proof}
We will show that the non-emptiness problem for unary $\xnfa$s is {\sf NP}-complete
which implies that the emptiness problem is {\sf coNP}-complete.
To show that the non-emptiness problem belongs to~{\sf NP} we use a similar approach
as described in Theorem~6.1 in~\cite{Stockmeyer:1973:WPR}. Let $M$ be an $\xnfa$ over
a unary alphabet~$\{a\}$
with state set $Q=\{q_1, q_2, \ldots, q_n\}$, initial state~$q_1$, and transition function~$\delta$. 
By applying Theorem~\ref{theo:xnfatodfa:upper} we know that there exists an
equivalent $\dfa$ that has at most $3^n$ states.
It is clear that $L(M)$ is not empty if and only if $M$ accepts a word of length $m \le 3^n$.

Now, the idea is first to guess a length $m \le 3^n$
in ternary representation $m_1m_2 \cdots m_n$ and to check whether there is 
exactly one
path of length~$m$ in~$M$ leading from the initial state to an accepting state.
The latter can be realized by mapping the transition function of~$M$ to its corresponding 
adjacency matrix~$A_M$ where we set an entry $A_M[i,j]=1$ if and only if
$q_j\in \delta(q_i,a)$, 
for $1 \le i,j \le n$. Then, $a^m \in L(M)$ if and only if the first row of $A_M^m$ has 
exactly one entry corresponding to an accepting state with value~$1$. 
Thus, we have as second task to compute the matrix product 
\mbox{$A_M^m=A_M^{m_1 \cdot 3^{n-1}} \cdot A_M^{m_2 \cdot 3^{n-2}}\cdots \cdot A_M^{m_n}$}
by inspecting the ternary counter. The matrix $A_M^m$ can be computed
by successively cubing and multiplying $A_M$. For example, let $m=22$ and its ternary notation be $211$.
Then, we have to multiply $A_M \cdot A_M^3 \cdot A_M^9 \cdot A_M^9$. In general, we
have at most \mbox{$3 \cdot 2 \log_3 (m) \le 6n$}
matrix multiplications. Since every matrix multiplication can be 
realized in time $n^2$, we obtain that $A_M^m$ can be computed in deterministic time bounded by a polynomial in~$n$.
Finally, the first row of the resulting matrix~$A_M^m$ has to be inspected. 
Altogether, these three tasks can be realized in nondeterministic time bounded by a polynomial in~$n$.
Hence, the complete procedure is in {\sf NP}.

To show that the non-emptiness problem is {\sf NP}-hard we use again a similar approach
as described in Theorem~6.1 in~\cite{Stockmeyer:1973:WPR}. It is shown there that a given Boolean formula
in conjunctive form with exactly three literals per conjunct is satisfiable if and only if a regular unary
language $L$ described by a regular expression is not equal to $\{a\}^*$. Moreover, the reduction is computable
in logarithmic space.
Since a language described by a regular expression can equivalently be described by an $\nfa$ of similar size,
we let now $L$ be described by an $\nfa$~$M$. Moreover, we construct a one-state $\dfa$ $M'$ that accepts $\{a\}^*$. 
Then, we construct an $\xnfa$ $M''$ that initially guesses whether it simulates for the complete input the 
$\nfa$ $M$ or the $\dfa$ $M'$. Since $M''$ is an $\xnfa$ we obtain that $L(M)=\{a\}^*$ if and only if $L(M'')=\emptyset$.
Hence, we have $L(M'') \neq \emptyset$ if and only if $L(M) \neq \{a\}^*$
if and only if the given Boolean formula is satisfiable.
Since the constructions of $M$, $M'$, and $M''$ can be realized in logarithmic space, we obtain the \textsf{NP}-hardness
of the non-emptiness problem for $\xnfa$s and, thus, the \textsf{coNP}-hardness of the emptiness problem for $\xnfa$s.
\end{proof}

\begin{theorem}\label{theo:compuniv:unary}
The problems of testing universality, 
inclusion, and equivalence for unary $\xnfa$s 
are {\em \textsf{coNP}}-complete.
\end{theorem}

\begin{proof}
Let us first show that the problems of testing non-universality, 
non-inclusion, and non-equivalence for unary $\xnfa$s are in {\sf NP}.
We start with the non-universality problem. Let $M$ be an $n$-state $\xnfa$.
By applying Theorem~\ref{theo:xnfatodfa:upper} we know that there exists an
equivalent $\dfa$ that has at most $3^n$ states. 
Hence, $L(M) \neq \{a\}^*$ if and only if there is a word of length $m \le 3^n$ that is not accepted by~$M$.
Similar to the proof of Theorem~\ref{theo:compempty:unary} we can guess a ternary representation of that word,
compute $A_M^m$, and check that the guessed word is not accepted by~$M$ by inspecting the first row
whether there is no entry corresponding to an accepting state with value~$1$. According to the
considerations made in the proof of Theorem~\ref{theo:compempty:unary} the procedure can be realized in
nondeterministic polynomial time and we obtain that the non-universality problem is in {\sf NP}. 
Hence, the universality problem is in {\sf coNP}.

Next, we consider the non-inclusion problem. Let $M_1$ be an $n_1$-state $\xnfa$ and $M_2$ be an $n_2$-state $\xnfa$.
By applying Theorem~\ref{theo:xnfatodfa:upper} we know that there exist equivalent $\dfa$s having at most
$3^{n_1}$ states and $3^{n_2}$ states, respectively. 
Hence, $L(M_1) \not\subseteq L(M_2)$ if and only if $L(M_1) \cap \overline{L(M_2)} \neq \emptyset$
if and only if
there is a word of length $m \le 3^{n_1+n_2}$ that is accepted by $M_1$, but not accepted by~$M_2$.
Similar to the proof of Theorem~\ref{theo:compempty:unary} and to the above construction for the non-universality
problem we obtain that the non-inclusion problem is in {\sf NP}. Hence, the inclusion problem is in {\sf coNP}.

Finally, we consider the equivalence problem. Let $M_1$ and $M_2$ be two $\xnfa$s.
Since the inclusion problem is in {\sf coNP}, we obtain that the equivalence problem is {\sf coNP}
by testing $L(M_1) \subseteq L(M_2)$ and $L(M_2) \subseteq L(M_1)$.

To show the \textsf{coNP}-hardness of the problems
we shortly describe how the reduction given in the proof of
Theorem~\ref{theo:compempty:unary} has to be extended.
We recall that we have constructed an $\xnfa$ $M''$
such that $L(M'') \neq \emptyset$ if and only if the given Boolean formula is satisfiable.

For non-universality we construct another $\xnfa$ $A$ 
that initially guesses whether it simulates for the complete 
input the $\xnfa$ $M''$ or the $\dfa$ $M'$ accepting $\{a\}^*$.
Then, we have $L(A) \neq \{a\}^*$ if and only if $L(M'') \neq \emptyset$
and obtain the \textsf{NP}-hardness of non-universality.
For the equivalence problem we consider~$M'$ as an $\xnfa$
and have $L(A)=L(M')=\{a\}^*$ if and only if $L(M'')=\emptyset$,
which gives the \textsf{coNP}-hardness of the equivalence problem.
Finally, we have \mbox{$L(M') \subseteq L(A)$} if and only if 
$L(A)=L(M')$ if and only if $L(M'')=\emptyset$ and obtain
the \textsf{coNP}-hardness of the inclusion problem.
\end{proof}

The computational complexity results in the unary case are summarized in Table~\ref{tab:complexity2}. 
\begin{table}[ht]
\begin{center}
\begin{tabular}{|l|c|c|c|c|}\hline
&$\dfa$&$\nfa$&$\xnfa$&$\afa$\\ \hline\hline
membership&{\sf L}&{\sf NL}&{\sf NL}&{\sf P}\\ \hline
emptiness&{\sf L}&{\sf NL}&{\sf coNP}&{\sf PSPACE}\\ \hline
universality&{\sf L}&{\sf coNP}&{\sf coNP}&{\sf PSPACE}\\ \hline
inclusion&{\sf L}&{\sf coNP}&{\sf coNP}&{\sf PSPACE}\\ \hline
equivalence&{\sf L}&{\sf coNP}&{\sf coNP}&{\sf PSPACE}\\ \hline
\end{tabular}
\end{center}
\caption{Computational complexity results for the decidability problems in the unary case.
All problems are complete with respect to the complexity class indicated. The results
for $\xnfa$s are obtained in this paper. The remaining results and pointers to the
literature are summarized, for example, in the survey~\cite{holzer:2011:dccfa}.}
\label{tab:complexity2}
\end{table}

\end{document}